\documentclass[11pt,a4paper]{amsart}
\usepackage{amsmath,amsthm,amssymb,amscd,dsfont,mathrsfs}
\newcommand{\ud}{\mathrm{d}}

\newcommand{\cD}{{\mathcal D}}

\newcommand{\cN}{\mathcal{N}}

\newcommand{\rz}{{\mathbb R}}
\newcommand{\nz}{{\mathbb N}}

\DeclareMathOperator{\orr}{o}

\DeclareMathOperator{\supp}{supp}

\newcommand{\eins}{\mathds{1}}

\numberwithin{equation}{section}

\newtheorem{theorem}{Theorem}[section]

\newtheorem{prop}[theorem]{Proposition}
\newtheorem{cor}[theorem]{Corollary}

\theoremstyle{definition}

\numberwithin{equation}{section}
\begin{document}
\title[Attractive conical surfaces]{Attractive conical surfaces create infinitely many bound states}
\author[S. Egger]{Sebastian Egger}
\address[S. Egger]{Department of Mathematics, Technion-Israel Institute of Technology 629 Amado Building, Haifa 32000, Israel}
\email{egger@tx.technion.ac.il}
\author[J. Kerner]{Joachim Kerner}
\address[J. Kerner]{Department of Mathematics and Computer Science, FernUniversit\"{a}t in Hagen, 58084 Hagen, Germany}
\email{Joachim.Kerner@fernuni-hagen.de}
\author[K. Pankrashkin]{Konstantin Pankrashkin}
\address[K. Pankrashkin]{D\'epartement de math\'ematiques, Univ. Paris-Sud, CNRS, Universit\'e Paris-Saclay,
91405 Orsay Cedex, France}
\email{konstantin.pankrashkin@math.u-psud.fr}
\begin{abstract} 
In this paper we study spectral properties of a three- \linebreak dimensional Schr\"odinger operator $-\Delta+V$ with a potential $V$ given, modulo rapidly decaying terms, by a function of the distance of $x \in \mathbb{R}^3$ to an infinite conical hypersurface
with a smooth cross-section. As a main result we show that there are infinitely many discrete eigenvalues accumulating at the bottom of the essential spectrum which itself is identified as the ground-state energy of a certain one-dimensional operator. Most importantly, based on a result of Kirsch and Simon we are able to establish the asymptotic behavior
of the eigenvalue counting function using an explicit spectral-geometric quantity associated with the cross-section.
This shows a universal character
of some previous results on conical layers and $\delta$-potentials created by conical surfaces.
%
\end{abstract}
\maketitle

\section{Introduction}

\subsection{Motivation and previous works}
In this paper we introduce a new class of long-range potentials $V$ leading to infinitely many eigenvalues below the bottom of the essential spectrum for the Schr\"odinger operator $-\Delta+V$. In addition, we estimate the accumulation rate of the eigenvalues at the bottom of the essential spectrum in terms of a certain geometric quantity.
Our work is motivated by several previous papers by various authors which appeared during the last decades and which studied particular classes of interactions, and we show that
the spectral effects observed are of a generic nature and hold in a much more general setting. In order to proceed with a more detailed discussion, let us introduce some objects.

By a conical surface $S\subset\rz^3$ we mean a Lipschitz surface invariant under the dilations, i.e. $\lambda S=S$ for any $\lambda>0$. A conical surface $S$ is uniquely determined by its cross-section $\Sigma$ given by $\Sigma:=S\,\cap\, \mathbb{S}^2$, where $\mathbb{S}^2$ is the unit sphere centered at the origin in $\rz^3$, and $S$ is recovered from $\Sigma$ by $S:=\rz_+\Sigma$.
By $d_S:\rz^3\to \rz_+$ we denote the distance function to $S$, i.e. ``$d_S(x):=$ the distance from $x$ to $S$''. In the present paper we deal with the Schr\"odinger operator
$H=-\Delta+V$, where the potential $V$ writes as $V(x)=v\big(d_S(x)\big) + w(x)$ with $v$ being a one-dimensional potential and $w$ being small in a suitable sense.
We show that, under suitable assumptions on $S$, $v$ and $w$, the essential spectrum of $H$ covers a half-axis $[\varepsilon_0,+\infty)$, where $\varepsilon_0$ only depends on $v$,
while the discrete eigenvalues form an infinite sequence converging to $\varepsilon_0$ with a rate controlled by the geometry of $S$.

It seems that the above spectral effect was first considered by Exner and Tater for conical layers \cite{et}. Namely, for $a>0$ consider the domain
$\Omega_a:=\big\{x\in\rz^3: d_S(x)<a\big\}$, which is called the conical layer of width $2a$ around $S$, and let $A$ be the Dirichlet Laplacian
in $\Omega_a$. With a slight abuse of interpretation, the operator $A$ corresponds to the above operator $H$ with a hard-wall potential $v$ (setting $w=0$),
\[
v(x)=\begin{cases}
0, & |x|<a,\\
+\infty, & |x|\ge a.
\end{cases}
\]
While some first results for the rotationally invariant case (i.e. when $\Sigma$ is a circle) were already obtained in \cite{et,DOBR}, we prefer to cite
directly the result obtained in \cite{OBP18} for general smooth cross-sections $\Sigma$: the essential spectrum of $A$ is $[\frac{\pi^2}{4a^2},+\infty)$
and the discrete spectrum is infinite provided that $S$ is not a plane (i.e. that $\Sigma$ is not a great circle of the unit sphere). Moreover, if
$\mathcal{N}_{\lambda}(A)$ denotes the number of discrete eigenvalues of $A$ in $(-\infty,\lambda)$, then one has the asymptotics
\begin{equation}  \label{eq-nnaa}
\cN_{\frac{\pi^2}{4a^2}-E}(A)\simeq k_S |\log E| \text{ as } E\to 0^+,
\end{equation}
where $k_S>0$ is some geometric constant given in \eqref{eqks} below. The above result belongs to a large family of
works on bound states in curved structures, see. e.g. \cite{dek,ek,es,LR}.

Another class of potential in the literature corresponds to the above operator $H$ with $v:=-\alpha \delta$ with $\alpha>0$ and $\delta$ being the Dirac $\delta$-function, which is an eminent model for so-called zero-range interactions~\cite{AGHH,Ex}.
More precisely, let $B$ be the self-adjoint operator in $\rz^3$ defined by the bilinear form
\[
b[\varphi,\varphi]=\int_{\rz^3} |\nabla \varphi|^2\ud x - \alpha \int_S \varphi^2\ud s, \quad
\varphi\in H^1(\rz^3),
\]
with $\ud s$ being the surface measure on $S$. Some preliminary results on the rotationally invariant case appeared 
in \cite{BEL,LB}, and the final result
of  \cite{OBP18} states that if $\Sigma$ is a smooth loop different from a great circle, then the essential spectrum of $B$ is $[-\frac{\alpha^2}{4},+\infty)$, while the discrete
spectrum is infinite and the asymptotics
\begin{equation}  \label{eq-nnbb}
\cN_{-\frac{\alpha^2}{4}-E}(B)\simeq k_S |\log E| \text{ as } E\to 0^+,
\end{equation}
holds with the \emph{same} constant $k_S>0$ as in \eqref{eq-nnaa}. Hence, the apparently different operators $A$ and $B$, which correspond to very different potentials $v$, share the same asymptotic behavior of the eigenvalue counting function. (Some
related results on conical domains and $\delta$-potentials can also be found e.g. in \cite{BPP,EL,OBPP}.)
Hence, the principal objective of the present work is to show that the above spectral picture is rather universal
and holds for a large class of one-dimensional potentials $v$.

\subsection{Main result}
Let us now pass to precise formulations. To keep the notation as simple as possible, we will work with real-valued Hilbert spaces. A self-adjoint and semi-bounded operator will be denoted with capital letters, e.g. $A$, and we denote by $\cD(A)$ its domain. By the corresponding small letters we denote its associated bilinear form, e.g. $a$, and by
$\cD(a)$ we mean the domain of the associated bilinear form (which may be referred to as the form domain of $A$).
The bilinear form itself will be denoted as $a[\cdot,\cdot]$ (by $a[\cdot]$ the form), the spectrum, the discrete spectrum, and the essential spectrum of $A$ will be denoted by $\sigma(A)$, $\sigma_d(A)$
and $\sigma_\text{ess}(A)$, respectively.  Also, $\mathcal{N}_{\lambda}(A)$ denotes the usual eigenvalue counting function, counting all the eigenvalues smaller or equal to $\lambda \in \mathbb{R}$ with multiplicities.

Let us start with geometric objects. Throughout the paper, we assume that $S$ is a conical surface whose cross-section $\Sigma$
is a $C^4$-smooth loop on the unit sphere. To avoid dealing with degenerate cases we assume that $\Sigma$ is a not a great circle, i.e. that $S$
is not simply a plane. We denote by $\ell$ the length of $\Sigma$, by $\mathbb{T}=\mathbb{R}/\ell\mathbb{Z}$ the circle of length $\ell$
and take an arbitrary arc-length parametrization $\Gamma$ of $\Sigma$, i.e. a $C^4$ map $\Gamma:\mathbb{T}\to \rz^3$ with $\Gamma(\mathbb{T})=\Sigma$
and $|\Gamma'|=1$. This implies the normal vector $n:=\Gamma\times \Gamma^{\prime}$
and the geodesic curvature $\kappa=\langle \Gamma \times \Gamma^{\prime\prime}\,, \Gamma^{\prime}\rangle_{\rz^3}$.
The essential object which will turn out to dictate the asymptotics of the eigenvalue counting function near the threshold to the essential spectrum is a one-dimensional Schr\"odinger operator involving in its potential part the geodesic curvature. Specifically, we introduce the geometrically induced Schr\"odinger operator $\mathcal{K}_S$ on $L^2(\mathbb{T})$ through
\begin{equation*}
\begin{split}
&(\mathcal{K}_S\varphi)(s):=-\varphi''(s)-\frac{\kappa^2(s)}{4}\varphi(s)\,,\\
&\quad \mathcal{D}(K_S):=H^2(\mathbb{T})\ .
\end{split}
\end{equation*}
Then $\mathcal{K}_S$ is self-adjoint, semibounded from below with compact resolvent. Let $\lambda_j(\mathcal{K}_S)$ denote its $j$-th eigenvalue (being enumerated in increasing order and counting multiplicities) then the following quantity is defined,
\begin{equation} 
   \label{eqks}
k_S:=\frac{1}{2\pi}\sum_{j \in \mathbb{N}: \, \lambda_j(\mathcal{K}_S) < 0}\sqrt{-\lambda_j(\mathcal{K}_S)},
\end{equation}
and $k_S>0$ under the above assumptions on $S$, as discussed in \cite{OBP18}.

Furthermore, let $v:\rz\to \rz$ be a one-dimensional potential with the following properties:
\begin{enumerate}
	\item[(i)] $v \in L^{1}_\text{loc}(\mathbb{R})$, $v(x)=v(-x)$ and $\max\{-v,0\}\in L^{\infty}(\mathbb{R})$\, ,
	\item[(ii)] The one-dimensional operator $Q:=-\frac{\ud^2}{\ud x^2}+v$ in $L^2(\rz)$
	defined through its form
	\[
	\begin{split}
	&q[\varphi]:=\int_{\mathbb{R}} \big(|\varphi'|^2+v|\varphi|^2\big)\, \ud x\, ,\\
	&\quad\cD(q):=\Big\{\varphi\in H^1(\rz): \int_{\rz} v|\varphi|^2\, \ud x<+\infty\Big\}\, ,
	\end{split}
	\]
	is such that $\inf \sigma(Q)=\varepsilon_0$ is an isolated eigenvalue,
	\item[(iii)] there holds 
	\begin{equation}
	\label{vvv}
	\varepsilon_0 < \liminf_{x \rightarrow \infty} v(x):=v_{\infty}\, .
	\end{equation}
	\end{enumerate}
It is widely known that the above assumptions are satisfied if, for example, $\liminf_{x\to+\infty} v(x)=+\infty$
or if $v\in L^1(\rz)\cap L^\infty(\rz)$ with $\int_{\rz} v\,\ud x<0$.

Now, let $V\in L^{1}_\text{loc}(\rz^3)$ with $\max\{-V,0\}\in L^{\infty}(\rz^3)$ such that the function
\[
w:x\mapsto V(x)- v\big(d_S(x)\big)
\]
satisfies
\begin{equation}
\label{decay}
w(x)=\orr(|x|^{-2})\text{ as } |x|\rightarrow+\infty.
\end{equation}
Finally, consider the Sch\"odinger operator $H=-\Delta+V$
realized via its associated quadratic form $h$,
\begin{equation}
\label{quadh}
\begin{split}
&h[\varphi]=\int_{\mathbb{R}^3} \left(|\nabla \varphi|^2 +V |\varphi|^2\right)\, \ud x\,,\\
&\quad \cD(h)=\{\varphi \in H^1(\mathbb{R}^3): \ \int_{\mathbb{R}^3} V|\varphi|^2 \, \ud x < \infty \}\,. 
\end{split}
\end{equation}

Our main result is as follows:

\begin{theorem}\label{MainTheorem} The essential spectrum of $H$ is $[\varepsilon_0,\infty)$ and the discrete spectrum is infinite.
In particular, $\mathcal{N}_{\varepsilon_0-E}(H)\simeq k_S|\ln E|$ as $E \rightarrow 0^+$.
\end{theorem}
In fact, after a few additional preparations, the proof of Theorem~\ref{MainTheorem} appears to be very similar to the one
given in \cite{OBP18} for $\delta$-potentials. It is our main observation that the necessary technical ingredients, which are collected
in the next section (Section~\ref{sec-prep}), are still available under the above assumptions on the potentials $v$ and $V$.
The proof of \cite{OBP18} is then adapted to our setting in Sections~\ref{SectionLower} and \ref{SectionUpper},
and it is based on suitable changes of variables, domain decoupling and a repeated use of the min-max principle.

We remark that the assumptions on both $v$ on $V$ could certainly be relaxed, in particular, by replacing the semiboundedness from below by some weaker integrability-type conditions. This would require an additional rather technical discussion of semiboundedness issues for $Q$ and $H$, which we preferred to avoid in the present text. We further remark that the smoothness of the cross-section $\Sigma$ is of importance for the whole construction: in fact, even in the simplest Fichera model of a conical layer with a non-smooth cross-section, the discrete spectrum becomes at most finite as observed in~\cite{ieot}.

\section{Some preparations on one-dimensional operators}\label{sec-prep}

Let us collect some important technical ingredients for the proof of the main result. 

The analysis of the eigenvalue counting function will use the following statement from \cite{KS88}; note that $(x)_+:=\max\{x,0\}$. 
\begin{prop}\label{KirschSimon} Let $x_0 \in \mathbb{R}$ and $V:[x_0,\infty) \rightarrow \mathbb{R}$ be continuous
with $c:=\lim_{x \rightarrow \infty}x^2V(x) \in \mathbb{R}$. Then the operator $h=-\frac{\ud^2}{\ud x^2}-V$ in $L^2(x_0,\infty)$ with any self-adjoint boundary condition at $x_0$ satisfies 
	\begin{equation*}
	\mathcal{N}_{-E}(h)= \frac{1}{2\pi}\sqrt{\left(c-\tfrac{1}{4}\right)_+} |\ln E| +o\big(|\ln E|\big) \text{ as }\ E \rightarrow 0^+\, .
	\end{equation*} 
	\end{prop}

Furthermore, we will need a number of facts on the truncated version of the above one-dimensional operator $Q$.	To this aim,
introduce two one-dimensional Schr\"odinger operators $H_{L,D/N}$ defined on $L^2(-L,L)$ via their associated quadratic forms 
\begin{equation*}\label{DirichletLL}
\begin{split}
&h_{L,D/N}[f]:=\int_{-L}^{L} \big(|f^{\prime}|^2+v|f|^2\big)\, \ud x\,,\\
&\quad\cD(h_{L,D}):=\Big\{f\in H^1_0(-L,L): \int_{|x|<L} v f^2\ud x < \infty \Big\},\\
&\quad\cD(h_{L,N}):=\Big\{f\in H^1(-L,L):  \int_{|x|<L} v f^2\ud x < \infty \Big\}\,.
\end{split}
\end{equation*}
Both operators, $H_{L,D}$ and $H_{L,N}$, have purely discrete spectrum and we denote the $n$-th eigenvalue as $\lambda_n(H_{L,D})$ and $\lambda_n(H_{L,N})$, respectively; here $n \in \mathbb{N}$ and we count the eigenvalues according to multiplicities. In what follows we will collect some estimates for the eigenvalues $\lambda_n(H_{L,D/N})$ as $L$ becomes large. 
The estimates are not surprizing and just correspond to what is expected, but we are not aware of a suitable general presentation in the literature.

For a self-adjoint semibounded from below operator $A$ in a Hilbert space $\mathcal H$ with  associated bilinear form $a$ we define the quantitites
\begin{equation}\label{EquatonRQ}
	\mu_j(A)=\inf_{W_j \subset \cD[a]}\sup_{\varphi \in W_j: \varphi \neq 0}\frac{a[\varphi,\varphi]}{\|\varphi\|^2_{\mathcal H}}\, , \quad j\in\nz,
	\end{equation}
	where $W_j$ denotes a $j$-dimensional subspace. The sequence $\mu_j$ is non-decreasing in $j$ and has a number of other properties. It is elementary to see that
	for two operators $A$ and $B$ with the above properties one has
	\begin{equation}
	   \label{eqab}
	\mu_j(A\oplus B)\ge \min\big\{\mu_j(A),\,\mu_1(B)\big\}.
	\end{equation}
	The well-known min-max principle states the following: if $\mu_j(A)<\inf \sigma_\text{ess}(A)$, then $\mu_j(A)$ is the $j$th eigenvalue of $A$.

\begin{prop}\label{proplim1}
For sufficiently large $L>0$ we have $\lambda_1(H_{L,N})\le\varepsilon_0$.
\end{prop}

\begin{proof}
Let $\widetilde{H}_{L,N}$ be the operator on $L^{2}(|x|>L)$ associated with the quadratic form
\begin{gather*}
%
\widetilde{h}_{L}[\varphi]:=\int_{L<|x|<\infty}\big((\varphi')^2+v \varphi^2\big) \, \ud x\,,\\
\cD(\widetilde{h}_{L}):=\{\varphi\in H^1(|x|>L):\ \widetilde{h}_{L}[\varphi,\varphi] < \infty  \}\,.
\end{gather*}
The min-max principle then implies that
\begin{equation}
   \label{eps00}
\varepsilon_0=\inf\sigma(Q) \geq \inf \sigma(H_{L,N}\oplus \widetilde{H}_{L,N})\,.
\end{equation}
On the other hand,
$\inf \sigma(H_{L,N} \oplus \widetilde{H}_{L,N})=\min\big\{\lambda_1(H_{L,N}), \inf\sigma(\widetilde{H}_{L,N})\big\}$, and due to the assumption~\eqref{vvv} one has $\inf\sigma(\widetilde{H}_{L,N})>\varepsilon_0$ for sufficiently large $L$. It follows from \eqref{eps00} that $\varepsilon_0\ge\lambda_1(H_{L,N})$ for large $L$. 
\end{proof}
	
\begin{prop}\label{proplim}
For any $j\in\nz$ there holds
\[
\liminf_{L\to+\infty} \lambda_j(H_{L,D/N})\ge \min\big\{\mu_j(Q),v_\infty\big\}.
\]
\end{prop}

\begin{proof}
First remark that both $H_{L,D/N}$ are with compact resolvents, hence, $\lambda_j(H_{L,D/N})=\mu_j(H_{L,D/N})$ for any $j\in\nz$.

As each function from the form domain of $H_{L,D}$ extends by zero to a function from the form domain of $Q$, which preserves the $L^2$-norm,
the min-max principle implies
\begin{equation}
   \label{ineqdir}
 \lambda_j(H_{L,D})=\mu_j(H_{L,D})\ge \mu_j(Q) \text{ for all } j\in\nz, \, L>0.
\end{equation}

Now, consider the case of $H_{L,N}$. We pick a pair of smooth functions $\chi_1,\chi_2 \in C^{\infty}(-1,1)$ such that $\chi_1(x)=1$ for $|x| \leq 1/3$, $\chi_2(x)=1$ for $|x| \geq 2/3$ and $\chi^2_1+\chi^2_2=1$. We set	$\chi^L_j(x):=\chi_j\left(\frac{x}{L}\right)$, and a straightforward calculation shows that 
	\begin{equation*}\begin{split}
	h_{L,N}[\varphi]\geq &h_{L,D}[\chi^L_1\varphi]+h_{L,N}[\chi^L_2\varphi]-\frac{c}{L^2}\|\varphi\|^2_{L^2(-L,+L)}\ ,
	\end{split}
	\end{equation*}
	for some constant $c > 0$. Now, introduce the self-adjoint operator $\widetilde{H}_L$ on $L^2(\Omega)$, $\Omega:=(-L,-L/3) \cup (L/3,+L)$, being associated with the quadratic form 
	\begin{equation*}\begin{split}
&\widetilde{h}_L[\varphi]:=\int_{\Omega}\left(|\nabla \varphi|^2+v(x)|\varphi|^2\right)\, \ud x\,,\\
&\quad\cD[\widetilde{h}_L]:=\{\varphi \in H^1(\Omega):\ \widetilde{h}_L[\varphi] < \infty  \}\,.
	\end{split}
	\end{equation*}
We can then define the injective map
	\begin{equation*}
	J:\cD[h_{L,N}]\rightarrow \cD[h_{2L/3,D}] \oplus \cD[\widetilde{h}_L], \quad J\varphi=(\chi^L_{1}\varphi,\chi^L_{2}\varphi),
	\end{equation*}
	which is also isometric with respect to $L^2$-norms. Employing $J$, the min-max principle implies that 
	\begin{equation*}
	\mu_j(H_{L,N}) \geq \mu_j(H_{2L/3,D}\oplus \widetilde{H}_L)-\frac{c}{L^2}.
	\end{equation*}
	Using \eqref{eqab} and then \eqref{ineqdir} we obtain
	\begin{align*}
	\mu_j(H_{L,N}) &\geq \min\big\{\mu_j(H_{2L/3,D}), \mu_1(\widetilde{H}_L)\big\}-\frac{c}{L^2},\\
	&\geq \min\big\{\mu_j(Q), \mu_1(\widetilde{H}_L)\big\}-\frac{c}{L^2}
	\end{align*}
	
	We have $\mu_1(\widetilde{H}_L)\ge \inf_{x\in\Omega} v(x)$, hence, $\liminf_{L\to+\infty}\mu_1(\widetilde{H}_L)\ge v_\infty$
	and 
	\[
	\liminf_{L\to+\infty} \lambda_j (H_{L,N})=\liminf_{L\to+\infty} \mu_j (H_{L,N})\geq \min\big\{\mu_j(Q), v_\infty\big\}. \qedhere
	\]
\end{proof}

\begin{cor}\label{corlow}
There are $\delta>0$ and $L_0>0$ such that $\lambda_2(H_{L,D/N})> \varepsilon_0+\delta$ for all $L>L_0$
\end{cor}

\begin{proof}
As the first eigenvalue of $Q$ is simple, one has $\mu_2(Q)>\varepsilon_0$.
In addition, one has $v_\infty>\varepsilon_0$ by assumption (iii) on $v$, and the result follows
from Proposition~\ref{proplim}.
\end{proof}

Now we provide a uniform Agmon-type estimate for $\varphi_{L,N}$, where $\varphi_{L,N}$ is the (unique, normalized) real-valued ground state for $H_{L,N}$. 
\begin{prop}[Uniform Agmon-type estimate]\label{AgmonEstimate}
For any $\theta\in(0,1)$ there exists $R>0$ and $C>0$ such that
	\begin{equation*}
	\int_{-L}^{L} e^{2\theta\Phi(x)}\varphi_{L,N}(x)^2\, \ud x < C\, , \quad \int_{\rz} e^{2\theta\Phi(x)}\varphi_{0}(x)^2\, \ud x < C\,,
	\end{equation*}
	hold for all $L > R > 0$, where
	\begin{equation*}
	\Phi(x):=\begin{cases}
	0\,,& x < R\, , \\
	\displaystyle\int_{R}^{|x|}\sqrt{v(t)-\varepsilon_0}\, \ud t\,,& R \leq|x|\, .
	\end{cases}
	\end{equation*}
\end{prop}
\begin{proof}
During the proof we simply write $\varphi$ instead of $\varphi_{L,N}$.

Take a sufficiently large $R>0$ such that $v(x)> \varepsilon_0$ for a.e. $x\ge R$.
Then $\Phi \in L^{\infty}(-L,L)$ and $\big|\Phi^{\prime}\big| \leq \eins_{\{L>|x|> R\}}\,\sqrt{v-\varepsilon_0}$,
where $\eins_{A}$ stands for the indicator function of the set $A \subset \mathbb{R}$.

Let us first show that for any $\eta>0$ one has $e^{\eta\Phi}\varphi  \in \cD(h_{L,N})$.
We have $e^{\eta\Phi}\in L^\infty(-L,L)$, so $e^{\eta\Phi}\varphi \in L^2(-L,L)$, and
\[
\int_{-L}^L v (e^{\eta\Phi}\varphi)^2\, \ud x<\infty \quad \text{ due to } \quad \int_{-L}^L v \varphi^2\, \ud x<\infty\,.
\]
Furthermore, $(e^{\eta\Phi}\varphi)'=\eta\Phi' e^{\eta\Phi}\varphi+ e^{\eta\Phi}\varphi'$.
The second summand is in $L^2(-L,L)$ due to $\varphi\in H^1(-L,L)$, while the first term is finite
due to
\[
\int_{-L}^L(\Phi' e^{\eta\Phi}\varphi)^2\, \ud x \le e^{2\eta C_L} \int_{R<|x|<L}(v-\varepsilon_0)\varphi^2\, \ud x<\infty\,,
\quad C_L:=\|\Phi\|_{L^\infty(-L,L)}.
\]
This implies that $e^{\eta\Phi}\varphi \in H^1(-L,+L)$.

In a next step we compute
\begin{equation*}
	\begin{split}
	h_{L,N}[e^{\theta \Phi}\varphi]&=\int_{-L}^L \Big( \big(\theta \Phi^{\prime} e^{\theta\Phi}\varphi+e^{\theta\Phi}\varphi'\big)^2 +v (e^{\theta \Phi}\varphi)^2\Big)\, \ud x\\
    &=\theta^2 \int_{-L}^L(\Phi^{\prime})^2e^{2\theta \Phi}\varphi^2\, \ud x + \int_{-L}^L\Big((e^{2\theta \Phi}\varphi)' \varphi' +v \,e^{2\theta\Phi}\varphi\,\varphi\, \Big)\,\ud x\,.
\end{split}
\end{equation*}
Taking advantage of the representation theorem for quadratic forms, the second term  on the right-hand side rewrites as
\begin{multline*}
\int_{-L}^L\Big((e^{2\theta \Phi}\varphi)' \varphi' +v \,e^{2\theta\Phi}\varphi\,\varphi\, \Big)\,\ud x\, 
=h_{L,N}\big[e^{2\theta\Phi}\varphi,\varphi\big]\\
=\langle e^{2\theta\Phi}\varphi, H_{L,N}\varphi\rangle_{L^2(-L,L)}
= \lambda_{1}(H_{L,N}) \int_{-L}^L e^{2\theta\Phi} \varphi^2\ud x
\end{multline*}
which then yields 
\[
%
h_{L,N}[e^{\theta\Phi}\varphi]=\int_{-L}^L\Big( \theta^2 (\Phi')^2 + \lambda_1(H_{L,N}) \Big) e^{2\theta\Phi}\varphi^2\,\ud x.
\]
Using Proposition~\eqref{proplim1} we conclude that for large $L$ we have
\begin{equation}
\label{q01}
h_{L,N}[e^{\theta\Phi}\varphi]\le \int_{-L}^L\big( \theta^2 (\Phi')^2 + \varepsilon_0 \big) e^{2\theta\Phi}\varphi^2\,\ud x\,.
\end{equation}

Now, let us pick a $\delta>0$.
Recall that by Proposition~\ref{proplim} we may assume $R$ sufficiently large to have
\begin{equation}
  \label{eq-lll}
\lambda_1(H_{R,N})>\varepsilon_0-\delta.
\end{equation}
The min-max principle applied to $h_{R,N}$ gives, with the help of \eqref{eq-lll},
\begin{multline*}
\int_{-R}^R\Big( \big((e^{\theta \Phi}\varphi)'\big)^2+v(e^{\theta \Phi}\varphi)^2 \Big)\, \ud x\\
\geq \lambda_1(H_{R,N})\int_{-R}^R e^{2\theta\Phi}\varphi^2 \, \ud x
\geq (\varepsilon_0-\delta)\int_{-R}^R e^{2\theta\Phi}\varphi^2 \, \ud x.
\end{multline*}
%
%
Therefore,
\begin{equation}
\begin{split}
h_{L,N}[e^{\theta\Phi}\varphi]&=\bigg(\int_{-R}^R+ \int_{R<|x|<L}\bigg)
\Big( \big((e^{\theta\Phi}\varphi)'\big)^2+v(e^{\theta\Phi}\varphi)^2 \Big) \, \ud x\\
&\ge (\varepsilon_0-\delta) \int_{-R}^{R}e^{2\theta\Phi}\varphi^2 \, \ud x + \int_{R<|x|<L} ve^{2\theta\Phi}\varphi^2 \, \ud x.
%
\end{split}
   \label{eq123}
\end{equation}
By combining \eqref{q01} with \eqref{eq123} we arrive at
\begin{multline*}
\int_{-L}^L\Big( \theta^2 (\Phi')^2 + \varepsilon_0\Big) e^{2\theta\Phi}\varphi^2\,\ud x\\
\geq (\varepsilon_0-\delta)\int_{-R}^{R}e^{2\theta\Phi}\varphi^2 \, \ud x + \int_{R<|x|<L} ve^{2\theta\Phi}\varphi^2 \, \ud x.
\end{multline*}
As $\Phi'=0$ in $(-R,R)$ by construction, the last inequality transforms into
\begin{align*}
\delta \int_{-R}^{R}e^{2\theta\Phi}\varphi^2 \, \ud x
&\ge
\int_{R<|x|<L} \big(v-\varepsilon_0-\theta^2(\Phi')^2\big) e^{2\theta\Phi}\varphi^2 \, \ud x.
\end{align*}
Using $\big|\Phi^{\prime}\big| \leq \eins_{\{L>|x|> R\}}\,\sqrt{v-\varepsilon_0}$ 
 we arrive at
\begin{align*}
\delta\int_{-R}^{R} \varphi^2 \, \ud x &\ge
(1-\theta^2)\int_{R<|x|<L} (v-\varepsilon_0) e^{2\theta\Phi}\varphi^2 \, \ud x.
\end{align*}
We may fix $\delta>0$ sufficiently small to have $v_\infty-\varepsilon_0 >2\delta (1-\theta^2)^{-1}$,
then for large (but then fixed) $R$ one has $(1-\theta^2)\big(v(x)-\varepsilon_0\big)>\delta$ for a.e. $R<|x|<L$
and it follows from the preceding inequality that
\[
\int_{-R}^{R} \varphi^2 \, \ud x\ge \int_{R<|x|<L}  e^{2\theta\Phi}\varphi^2 \, \ud x\,.
\]
Consequently,
\begin{multline*}
\int_{-L}^{L}e^{2\theta\Phi}\varphi^2\, \ud x=\int_{-R}^{R}e^{2\theta\Phi}\varphi^2\, \ud x
+\int_{R<|x|<L}e^{2\theta\Phi}\varphi^2\, \ud x\\
=\int_{-R}^{R}\varphi^2\, \ud x+\int_{R<|x|<L}e^{2\theta\Phi}\varphi^2\, \ud x 
\le 2\int_{-R}^{R}\varphi^2\ud x. \qedhere
\end{multline*}
\end{proof}
\begin{cor}
\label{expdecay}
For every $\eta>0$ and $0<\theta<1$ there exist $B>0$ and $b>0$ such that 
\begin{equation*}
\|\varphi_{L,N}\|_{L^2(L-\eta<|x|<L)}\le B e^{-b L}\|\varphi_{L,N}\|_{L^2(-L,L)}
\end{equation*}
as $L$ is sufficiently large.
\end{cor}
\begin{proof}
Starting with the definition of $\Phi$ in Proposition~\ref{AgmonEstimate} we see that
$\Phi(x)>\delta(|x|-R)$ for some $\delta>0$ and $|x|>R$. Proposition \ref{AgmonEstimate} then yields
\begin{equation*}
\begin{split}
\|\varphi_{L,N}\|_{L^2(L-\eta<|x|<L)}&\leq\|e^{-\theta\Phi}e^{\theta\Phi}\varphi_{L,N}\|_{L^2(L-\eta<|x|<L)}\\
&\leq e^{-\theta\delta(L-\eta)}\|e^{\theta\Phi}\varphi_{L,N}\|_{L^2(L-\eta<|x|<L)}\\
&\leq e^{-\theta\delta(L-\eta)}\|e^{\theta\Phi}\varphi_{L,N}\|_{L^2(-L,L)}\\
&\leq C e^{-\theta\delta L} \|\varphi_{L,N}\|_{L^2(-L,L)}. \qedhere
\end{split}
\end{equation*}
\end{proof}
The next result controls the convergence of the eigenvalues $\lambda_1(H_{L,N})$ to $\varepsilon_0$.
\begin{prop}\label{propneu} With some $A>0$ and $a>0$ one has 
	\begin{equation}
	\varepsilon_0-A\mathrm{e}^{-a L} \leq \lambda_1(H_{L,N}) \leq \varepsilon_0. 
	\end{equation}
	holds for large $L$.
	\end{prop}
	\begin{proof}
	The upper bound is already shown in Proposition~\ref{proplim1}.
	To prove the lower bound, let us take  $\chi\in C^{\infty}(0<|x|<1)$ with $\chi(x)=1$ for $0<|x|<\frac{1}{3}$, $\chi(x)=0$ for $\frac{2}{3}<|x|<1$ and $|\chi(x)| \leq 1$ for all $0 < x < 1$. We introduce $\xi_{1,L}\in C^{\infty}(-L,L)$ via
	\begin{equation}
	\label{phi1}
	\xi_{1,L}(x):=\begin{cases}
	1\ , & |x| \leq L-1\ ,\\
	\chi(|x|-L+1)\ , & L-1 < |x| < L\ .
	\end{cases}
	\end{equation}
	We also define $\xi_{2,L}:=1-\xi_{1,L}$ and set
	\begin{equation}
	\chi_{j,L}:=\frac{\xi_{j,L}}{\sqrt{\xi^2_{1,L}+\xi^2_{2,L}}}\ , \quad j=1,2\ .
	\end{equation}
	
	Note that $\|\chi_{j,L}'\|_{\infty}$ is uniformly bounded for large enough $L$, $j=1,2$, and $\supp(\chi_{1,L}'+\chi_{2,L}')\subset[L-\tfrac{2}{3},L-\tfrac{1}{3}]$. With this and Corollary~\ref{expdecay} we conclude that there are two constants $C_1,C_2>0$ such that 
	\begin{equation*}
	\begin{split}
	h_{L,N}[\varphi_{L,N}]&=h_{L,N}[\chi_{1,L}\varphi_{L,N}]+h_{L,N}[\chi_{2,L}\varphi_{L,N}]\\
	&\quad-\int_{-L}^{L}\left[((\chi_{1,L})^{\prime})^2+ ((\chi_{2,L})^{\prime})^2\right]
	\varphi_{L,N}^2\, \ud x \\
	&=h_{L,N}[\chi_{1,L}\varphi_{L,N}]+h_{L,N}[\chi_{2,L}\varphi_{L,N}]\\
	&\quad -\int_{L-\frac{2}{3}<|x|<L}\left[((\chi_{1,L})^{\prime})^2+ ((\chi_{2,L})^{\prime})^2\right]
	\varphi_{L,N}^2\, \ud x \\
	&\geq h_{L,N}[\chi_{1,L}\varphi_{L,N}]+h_{L,N}[\chi_{2,L}\varphi_{L,N}] -C_2e^{-C_1L}\|\varphi_{L,N}\|^2_{L^2(-L,L)}\ .
	\end{split}
	\end{equation*}
	Since $(\chi_{1,L}\varphi_{L,N})(x)=0$ for $|x|\ge L-\frac{1}{3}$ we conclude that
	$\chi_{1,L}\varphi_{L,N}\in \cD(q)$ and 
	\begin{equation*}\begin{split}
	h_{L,N}[\chi_{1,L}\varphi_{L,N}]=q[\chi_{1,L}\varphi_{L,N}]\ge\varepsilon_0 \|\chi_{1,L}\varphi_{L,N}\|^2_{L^2(-L,L)}\,.
	\end{split}
	\end{equation*}
	Using \eqref{vvv}, one can choose $L$ sufficiently large to have
	$v> \varepsilon_0$ for $|x|>L-\tfrac{2}{3}$. Since {$\supp \left(\chi_2\varphi_{L,N}\right)\subset[L-\tfrac{1}{3},L]$}
	there holds
	\begin{multline*}
	h_{L,N}[\chi_{2,L}\varphi_{L,N}]\geq\int_{L-\frac{1}{3}}^L v \,(\chi_{2,L}\varphi_{L,N})^2\, \ud x \\
	\geq \varepsilon_0 \int_{L-\frac{1}{3}}^L (\chi_{2,L}\varphi_{L,N})^2\, \ud x =\varepsilon_0 \|\chi_{2,L}\varphi_{L,N}\|^2_{L^2(-L,L)}\,.
	\end{multline*}
Therefore, for large enough $L$ we have 
\begin{equation*}\begin{split}
\lambda_{1}(H_{L,N})\|\varphi_{L,N}\|^2_{L^2(-L,L)}&=h_{L,N}[\varphi_{L,N}]\\
&\geq\varepsilon_0 \|\chi_{1,L}\varphi_{L,N}\|^2_{L^2(-L,L)}+\varepsilon_0 \|\chi_{2,L}\varphi_{L,N}\|^2_{L^2(-L,L)}\\
& \qquad -C_2e^{-C_1L}\|\varphi_{L,N}\|_{L^2(-L,L)}\\
&=\big(\varepsilon_0-C_2e^{-C_1L}\big)\|\varphi_{L,N}\|^2_{L^2(-L,L)}\,. \qedhere
\end{split}
\end{equation*}
\end{proof}

\begin{prop}\label{propdir} There exist $A>0$ and $a>0$ such that 
	\begin{equation}
	\varepsilon_0\leq \lambda_1(H_{L,D}) \leq \varepsilon_0+A\mathrm{e}^{-a L}. 
	\end{equation}
	holds for large $L$.
	\end{prop}

\begin{proof}
The lower bound is an immediate consequence of the min-max principle.
For the upper bound we use the same functions $\chi_{j,L}$ as in the proof of Proposition~\ref{propneu}.
Applying Corollary \ref{expdecay} we estimate, with some $B>0$ and $b>0$,
		\begin{equation*}
		\begin{split}
		h_{L,N}[\varphi_{L,N}]
		&=h_{L,N}[\chi_{1,L}\varphi_{L,N}]+h_{L,N}[\chi_{2,L}\varphi_{L,N}]\\
		&\qquad-\int_{L-\frac{2}{3}<|x|<L-\frac{1}{3}}\left[(\chi'_{1,L})^2+ (\chi'_{2,L})^2\right]
	\varphi_{L,N}^2\, \ud x\\
	&\geq h_{L,D}[\chi_{1,L}\varphi_{L,N}]+h_{L,N}[\chi_{2,L}\varphi_{L,N}]
		-Be^{-bL}\|\varphi_{L,N}\|^2_{L^2(-L,L)}
		\end{split}
		\end{equation*}
By the min-max principle we have
\[
h_{L,D}[\chi_{1,L}\varphi_{L,N}]\ge \lambda_1(H_{L,D})\|\chi_{1,L}\varphi_{L,N}\|^2_{L^2(-L,L)}.
\]
Furthermore, as $\supp \chi_{2,L}\varphi_{L,N}\subset (L-1,L)$, using the assumption (iii) on $v$ we obtain
\[
h_{L,N}[\chi_{2,L}\varphi_{L,N}]\ge \inf_{|x|>L-1} v(x) \|\chi_{1,L}\varphi_{L,N}\|^2_{L^2(-L,L)}
\ge \varepsilon_0 \|\chi_{1,L}\varphi_{L,N}\|^2_{L^2(-L,L)}.
\]
Alltogether the preceding estimates give
\begin{align*}
\lambda_1(H_{L,N})\|\varphi_{L,N}\|^2_{L^2(-L,L)}
&=h_{L,N}[\varphi_{L,N}]\\
&\ge  \lambda_1(H_{L,D})\|\chi_{1,L}\varphi_{L,N}\|^2_{L^2(-L,L)}+\varepsilon_0\|\chi_{2,L}\varphi_{L,N}\|^2_{L^2(-L,L)}\\
&\qquad -Be^{-bL}\|\varphi_{L,N}\|^2_{L^2(-L,L)},
\end{align*}
which results in
\begin{multline*}
\lambda_1(H_{L,D})\|\chi_{1,L}\varphi_{L,N}\|^2_{L^2(-L,L)}\\
\le 
\big(\lambda_1(H_{L,N})+Be^{-bL}\big)\|\varphi_{L,N}\|^2_{L^2(-L,L)}
-\varepsilon_0\|\chi_{2,L}\varphi_{L,N}\|^2_{L^2(-L,L)}
\end{multline*}
and then
\[
\lambda_1(H_{L,D})\le
\big(\lambda_1(H_{L,N})+B e^{-bL}\big) \dfrac{\|\varphi_{L,N}\|^2_{L^2(-L,L)}}{\|\chi_{1,L}\varphi_{L,N}\|^2_{L^2(-L,L)}}
-\varepsilon_0 \dfrac{\|\chi_{2,L}\varphi_{L,N}\|^2_{L^2(-L,L)}}{\|\chi_{1,L}\varphi_{L,N}\|^2_{L^2(-L,L)}}
\]
With some $c>0$ we have $\lambda_1(H_{L,N})=\varepsilon_0+O(e^{-cL})$ by Proposition~\ref{propneu},
and 
\begin{align*}
\|\chi_{1,L}\varphi_{L,N}\|^2_{L^2(-L,L)}&=\big(1+O(e^{-cL})\big)\|\varphi_{L,N}\|^2_{L^2(-L,L)},\\
\|\chi_{2,L}\varphi_{L,N}\|^2_{L^2(-L,L)}&= O(e^{-cL}) \|\varphi_{L,N}\|^2_{L^2(-L,L)},
\end{align*}
by Corollary \ref{expdecay}, which gives the sought estimate.
\end{proof}

\section{Proof of Theorem~\ref{MainTheorem}: Lower bound }\label{SectionLower}
With the preparation from the preceding section, the proofs for both upper and lower bounds
follow the same steps as in \cite[Sec.~3]{OBP18}. Nevertheless we recall the general scheme.

In a first step we note that a classical result of differential geometry allows us to find $R_0,\delta_0 > 0$, $\delta_0\in(0,\kappa_{\infty}^{-1})$, such that, for all $R > R_0$ and $\delta \in (0,\delta_0)$,
\begin{equation}\label{TubularCoordinates}
\begin{split}
&\Phi: \Pi_{R,\delta}:= (R,+\infty) \times \mathbb{T} \times (-\delta R ,+\delta R) \rightarrow \rz^3\,,\\
& \quad (r,s,t) \mapsto r\Gamma(s)+tn(s) 
\end{split}
\end{equation}
is injective such that $d_S(\Phi(r,s,t))=|t|$. On the open set $\Omega_R:=\Phi(\Pi_{R,\delta})$ we define the quadratic form, for some $c_R>0$ suitably large, 
\begin{equation}
\label{forma}
\begin{split}
&a_{R,D}[\varphi]:=\int_{\Omega_R} \left(|\nabla \varphi|^2+V |\varphi|^2\right)\, \ud x\,,\\
&\quad\cD\left(a_{R,D}\right):=\{ \varphi \in H^1_0(\Omega_R): \ \sqrt{|V|}\varphi\in L^2(\Omega_R)\}\,.
\end{split}
\end{equation}
By an operator bracketing argument one concludes that $\mathcal{N}_{\varepsilon_0 -E}(A_{R,D}) \leq \mathcal{N}_{\varepsilon_0 -E}(H)$.

Let $\varepsilon>0$, then by \eqref{decay} we can choose $R$ large enough to estimate
\[
V(x)\le v\big(d_S(x)\big)+ \dfrac{\varepsilon}{|x|^2} \text{ for } x\in\Omega_R.
\]
Hence the form $a_{R,D}$ is bounded from above by the form $b_{R,D}$,
\begin{equation}
\label{formb}
\begin{split}
&b_{R,D}[\varphi]:=\int_{\Omega_R} \left(|\nabla \varphi|^2+ \Big(v\big(d_S(x)\big)+ \dfrac{\varepsilon}{|x|^2}\Big)|\varphi|^2\right)\, \ud x\,,\\
&\quad\cD\left(b_{R,D}\right):=\{ \varphi \in H^1_0(\Omega_R): \ \sqrt{|v\circ d_S|}\varphi\in L^2(\Omega_R)\}\,,
\end{split}
\end{equation}
and then $\mathcal{N}_{\varepsilon_0 -E}(B_{R,D})\le \mathcal{N}_{\varepsilon_0 -E}(A_{R,D})$, which then implies
$\mathcal{N}_{\varepsilon_0 -E}(B_{R,D})\le \mathcal{N}_{\varepsilon_0 -E}(H)$.

In a next step one employs the diffeomorphism $\Phi$ in \eqref{TubularCoordinates}
and the associated unitary map
\[
U: L^2(\Omega_R)\to L^2(\Pi_{R,\delta}),
U\varphi:=\sqrt{\det \Phi'} \varphi\circ \Phi.
\]
A standard computation (change of variables) shows that the form $c_{R,\delta}$, $c_{R,\delta}[U\varphi]:=b_{R,D}[\varphi]$, is given by
\begin{equation}\label{EqUnitaryForm}
\begin{split}
c_{R,\delta}[\varphi]:&=\int_{\Pi_{R,\delta}} |\partial_r \varphi|^2 +\frac{1}{(r+t\kappa)^2}\left(|\partial_s \varphi|^2-\frac{\kappa^2+1}{4} |\varphi|^2 \right)+|\partial_t \varphi|^2\\
&\ +\left(v(t)+\frac{\varepsilon}{r^2+t^2}\right) |\varphi|^2+\left(\frac{t\kappa^{\prime \prime}}{2(r+t\kappa)^3}-\frac{5}{4}\frac{(t\kappa^{\prime})^2}{(r+t\kappa)^4}\right) |\varphi|^2 \, \ud r \, \ud s \ud t\, ,\\
\cD(c_{R,\delta}):&= \big\{\varphi \in L^2(\Pi_{R,\delta}): \ \partial_r\varphi,\, r^{-1}\partial_s\varphi,\,\partial_t\varphi,\,\sqrt{|v(t)|}\varphi \in L^2(\Pi_{R,\delta}),\\
&\qquad \ \varphi=0 \ \text{on} \ \partial \Pi_{R,\delta} \big\}.
\end{split}
\end{equation}
%
As a result, $\mathcal{N}_{\varepsilon_0 -E}(C_{R,\delta})\le \mathcal{N}_{\varepsilon_0 -E}(H)$.

Then, in addition, one introduces $\widehat{\Pi}_{R,\delta}:=(1,+\infty) \times \mathbb{T} \times (-\delta R,+\delta R)$ and a coordinate transformation 
\begin{equation}
\label{wphi}
\begin{split}
&\widehat{\Phi}:\Pi_{R,\delta}\rightarrow\widehat{\Pi}_{R, \delta}\,,\\
&\quad (r,s,t)\rightarrow (\rho,s,t)\,, \quad \rho=\frac{r-\delta R \kappa_{\infty}}{R(1-\delta \kappa_{\infty})}\,.
\end{split}
\end{equation}
Notice that $\rho\geq 1$ by the choice of $\delta_0$ in the beginning of Section \ref{SectionLower}.
Then rewriting the form $c_{R,\delta}$ in the new coordinates and applying simple upper bounds by controlling
$\kappa^{(j)}$ through their sup-norms (the computations
of \cite[Sec.~3]{OBP18} apply almost literally here) we show that
$\mathcal{N}_{\varepsilon_0 -E}(G_{R,\delta})\leq \mathcal{N}_{\varepsilon_0 -E}(C_{R,\delta})$,
where $G_{R,\delta}$ is the operator in $L^2(\widehat{\Pi}_{R,\delta})$ generated by the quadratic form $g_{R,\delta}$,
\begin{align*}
g_{R,\delta}[\varphi]&:=\int_{\widehat{\Pi}_{R,\delta}} |\partial_t \varphi|^2+v(t) |\varphi|^2+\frac{1}{R^2(1-\delta\kappa_{\infty})^2}|\partial_\rho \varphi|^2 \\
&\quad +\frac{1}{R^2(1-\delta\kappa_{\infty})^2\rho^2}\left(|\partial_s \varphi|^2-\frac{\kappa^2+1-C(\delta+\varepsilon)}{4} |\varphi|^2 \right)
 \, \ud \rho \, \ud s \, \ud t\,,\\
\cD(g_{R,\delta})&:=\{\varphi \in L^2(\widehat{\Pi}_{R,\delta}): \partial_{\rho}\varphi, \rho^{-1}\partial_s\varphi,\partial_t\varphi, \sqrt{|v(t)|}\varphi \in L^2(\Pi_{R,\delta}),\\
&\qquad \ \varphi=0 \ \text{on}\ \partial \widehat{\Pi}_{R,\delta} \}\,. 
\end{align*}
with $C> 0$ being some fixed constant. It follows 
\[
\mathcal{N}_{\varepsilon_0 -E}(G_{R,\delta})\leq \mathcal{N}_{\varepsilon_0 -E}(H),
\]
and hence it is sufficient to show that
\begin{equation*}
\liminf_{E \rightarrow 0^+}\frac{\mathcal{N}_{\varepsilon_0 -E}(G_{R,\delta})}{|\ln E|} \geq k_S\, .
\end{equation*}
It is easy to see that $G_{R,\delta}$ commutes with $\eins\otimes(\mathcal{K}_S\otimes\eins)$ and $\eins\otimes(\eins\otimes H_{\delta R,D})$ and one identifies $L^2(\widehat{\Pi}_{R,\delta})\simeq L^2(1,+\infty)\otimes L^2(\mathbb{T})\otimes L^2(-\delta R,\delta R)$. The decomposition with respect to the eigenbases of $\mathcal{K}_S$
and $H_{\delta R,D}$ implies the unitary equivalence  
\begin{equation*}
G_{R,\delta} \simeq \bigoplus_{m,n \in \mathbb{N}}\left(\frac{1}{R^2(1-\delta\kappa_{\infty})^2}G^{[m]}_{R,\delta}+\lambda_n(H_{\delta R,D})  \right)\, ,
\end{equation*}
where $\lambda_n(H_{\delta R,D})$ is the $n$-th eigenvalue of the operator $H_{\delta R,D}$. Also, $G^{[m]}_{R,\delta}$ is the operator associated with the quadratic form 
\begin{equation*}
\begin{split}
&g^{[m]}_{R,\delta}[f]:=\int_{1}^{+\infty}\left(|f^{\prime}|^2+\frac{\lambda_m(\mathcal{K}_S)-\frac{1-C(\delta+\varepsilon)}{4}}{\rho^2}|f|^2  \right)\, \ud \rho\,,\\
&\quad\cD(g^{[m]}_{R,\delta}):=H^1_0(1,+\infty)\,. 
\end{split}
\end{equation*}
Setting 
\begin{equation*}
M_{\delta}:=\max\{m\in \mathbb{N}: \ \lambda_m(\mathcal{K}_S)-\frac{1-C(\delta+\varepsilon)}{4}< 0   \}
\end{equation*}
we estimate 
\begin{equation}
\label{ineqn}\begin{split}
\mathcal{N}_{\varepsilon_0 -E}(G_{R,\delta})&=\sum_{m,n\in \mathbb{N}}\mathcal{N}_{\varepsilon_0 -E}\left(\frac{1}{R^2(1-\delta\kappa_{\infty})^2}G^{[m]}_{R,\delta}+\lambda_n(H_{\delta R,D})  \right) \\
&\geq \sum_{m=1}^{M_{\delta}}\mathcal{N}_{\varepsilon_0 -E}\left(\frac{1}{R^2(1-\delta\kappa_{\infty})^2}G^{[m]}_{R,\delta}+\lambda_1(H_{\delta R,D})  \right) \\
&=\sum_{m=1}^{M_{\delta}}\mathcal{N}_{\varepsilon_0 -E-\lambda_1(H_{\delta R,D})}\left(\frac{1}{R^2(1-\delta\kappa_{\infty})^2}G^{[m]}_{R,\delta} \right)\\
&=\sum_{m=1}^{M_{\delta}}\mathcal{N}_{-\mu_\delta(E)}(G^{[m]}_{R,\delta}),
\end{split}
\end{equation}
where we denote
\[
\mu_{\delta}(E):=\left(E-\varepsilon_0+\lambda_1(H_{\delta R, D})  \right)R_{\delta}(E)^2(1-\delta\kappa_{\infty})^2\, .
\]
Up to now, $R> 0$ was a fixed constant. It is now customary to make it actually dependent on $E$ such that $R=R_{\delta}(E):=K_{\delta}|\ln E|$ with $K_\delta > 0$, some constant that will be specified shortly. By Proposition~\ref{propdir} we know that 
\begin{equation*}
\left|\varepsilon_0-\lambda_1(H_{\delta R_{\delta}(E), D})   \right| \leq C^{-1}_0\mathrm{e}^{-C_0K_{\delta}\delta|\ln E|}=C^{-1}_0E^{C_0K_{\delta}\delta}\, .
\end{equation*}
Hence, choosing $K_{\delta} > 0$ large enough, we conclude that
\begin{equation*}
\mu_{\delta}(E)=(1-\delta\kappa_{\infty})^2K^2_{\delta}|\ln E|^2E+\orr(E|\ln E|^2) \ \text{for}\ E \rightarrow 0^+\, .
\end{equation*}
Applying Proposition~\ref{KirschSimon} to each summand in the last term of \eqref{ineqn} we obtain
\begin{equation*}\begin{split}
\liminf_{E \rightarrow 0^+}\frac{\mathcal{N}_{\varepsilon_0 -E}(G_{R,\delta})}{|\ln E|} &\geq \sum_{m=1}^{M_{\delta}} \lim_{E \rightarrow 0^+}\frac{\mathcal{N}_{-\mu_{\delta}(E)}(G^{[m]}_{R,\delta})}{|\ln E|} \\
&=\sum_{m=1}^{M_{\delta}} \lim_{E \rightarrow 0^+}\left(\frac{\mathcal{N}_{-\mu_{\delta}(E)}(G^{[m]}_{R,\delta})}{|\ln \mu_{\delta}(E)|} \frac{|\ln \mu_{\delta}(E)|}{|\ln E|}\right) \\
&=\sum_{m=1}^{M_{\delta}} \lim_{E \rightarrow 0^+}\frac{\mathcal{N}_{-\mu_{\delta}(E)}(G^{[m]}_{R,\delta})}{|\ln \mu_{\delta}(E)|} \\
&=\frac{1}{2\pi}\sum_{m=1}^{M_{\delta}} \sqrt{\left(\lambda_m(\mathcal{K}_S)-\frac{C}{4}(\delta+\varepsilon)\right)_+}\, .
\end{split}
\end{equation*}
Now, the above inequality implies
\begin{equation*}\begin{split}
\liminf_{E \rightarrow 0^+}\frac{\mathcal{N}_{\varepsilon_0 -E}(H)}{|\ln E|} &\geq \frac{1}{2\pi}\sum_{m=1}^{M_{\delta}} \sqrt{\left(\lambda_m(\mathcal{K}_S)-\frac{C}{4}(\delta+\varepsilon)\right)_+}.
\end{split}
\end{equation*}
Since both $\delta$ and $\varepsilon$ can be chosen arbitrarily small,
the result follows.
\section{Proof of Theorem~\ref{MainTheorem}: Upper bound} \label{SectionUpper}
In this section we aim to prove that
\begin{equation*}
\limsup_{E \rightarrow 0^+}\frac{\mathcal{N}_{\varepsilon_0-E}(H)}{|\ln E|} \leq k_S\, .
\end{equation*}
In order to do this we will work on a different neighborhood of $S$. More explicitly, we introduce 
\begin{equation*}\begin{split}
\mathcal{P}_{R,\delta}:&=\big\{(r,t) \in \mathbb{R}^2: r > R,\ t \in (-\delta r,+\delta r)  \big\} \\
&=\big\{(r,t) \in \mathbb{R}^2: r > r_{R,\delta}(t)  \big\}\,, \quad
r_{R,\delta}(t):=\max\left(R,\frac{|t|}{\delta}  \right).
\end{split}
\end{equation*}
We set $\mathcal{V}_{R,\delta}:=\mathcal{P}_{R,\delta} \times \mathbb{T}$. Then, there exists $R_0 > 0$ and $\delta_0 \in (0,\kappa^{-1}_{\infty})$ such that the map 
\begin{equation}\label{LambdaMap}
\begin{split}
&\Lambda:\mathcal{V}_{R,\delta} \rightarrow \mathbb{R}^3\,,\\
& \quad \Lambda(r,s,t)\rightarrow r\Gamma(s)+tn(s)
\end{split}
\end{equation}
is injective for all $\delta \in (0,\delta_0)$ and $R > R_0$. On the domain $\Omega_{R,\delta}:=\Lambda(\mathcal{V}_{R,\delta})$ we then introduce the quadratic form
\begin{equation*}
\begin{split}
&a_{R,\delta}[\varphi]=\int_{\Omega_{R,\delta}}\big(|\nabla \varphi|^2+ V|\varphi|^2\big)\, \ud x\,,\\
&\quad\cD(a_{R,\delta})=\{\varphi \in H^1(\Omega_{R,\delta}):\int_{\Omega_{R,\delta}} V|\varphi|^2\, \ud x < \infty\}\,. 
\end{split}
\end{equation*}
In order to show that $a_{R,\delta}$ can be used to obtain the upper bound for the eigenvalue counting function
one uses the standard operator bracketing argument. Namely, choose $R_1>R$ in such a way that
$V(x)>\varepsilon$ for $x\notin B(0,R_1)\,\cup \Omega_{R,\delta}$.
We then denote $\Omega_0:=B(0,R_1)\setminus \Omega_{R,\delta}$ and $\Omega_1:=\rz^2\setminus \overline{B(0,R_1)\,\cup \Omega_{R,\delta}}$ and consider the quadratic form
\begin{equation*}
\begin{split}
&\Hat h[\varphi]=\int_{\Omega_{R,\delta}\cup \Omega_0\cup\Omega_1}\big(|\nabla \varphi|^2+ V|\varphi|^2\big)\, \ud x\,,\\
&\quad\cD(\Hat h_{R,\delta})=\{\varphi \in H^1(\Omega_{R,\delta}\cup \Omega_0\cup\Omega_1):\int_{\Omega_{R,\delta}\cup \Omega_0\cup\Omega_1} V|\varphi|^2\, \ud x < \infty\},
\end{split}
\end{equation*}
which is an extension of the quadratic form for $H$. Then a standard application of the minimax principle shows that
$\cN_{\varepsilon_0-E}(H)\le \cN_{\varepsilon_0-E}(\Hat H)$ for any $E>0$. We further remark that
$\Hat H$ represents as the direct sum of three operators acting in $L^2(\Omega_{R,\delta})$, $L^2(\Omega_0)$, $L^2(\Omega_1)$,
$\Hat H=A_{R,\delta}\oplus A_0\oplus A_1$, hence, 
\[
\cN_{\varepsilon_0-E}(\Hat H)=\cN_{\varepsilon_0-E}(A_{R,\delta})+\cN_{\varepsilon_0-E}(A_0)+\cN_{\varepsilon_0-E}(A_1), \quad E>0.
\]
Due to the choice of $R_1$ one has $A_1\ge \varepsilon_0$, hence the last summand is zero.
The operator $A_1$ has compact resolvent, hence, for any $E>0$ we have $\cN_{\varepsilon_0-E}(A_0)\le \cN_{\varepsilon_0}(A_0)<\infty$. To summarize, for any $R>R_0$ and $\delta\in(0,\delta_0)$ there is $C>0$ such that
	$	\mathcal{N}_{\varepsilon_0-E}(H)\leq \mathcal{N}_{\varepsilon_0-E}(A_{R,\delta})+C$.

In addition, if one takes any $\varepsilon>0$, then by \eqref{decay} we can assume $R_0$ large enough to estimate
\[
V(x)\ge v\big(d_S(x)\big)- \dfrac{\varepsilon}{|x|^2} \text{ for } x\in\Omega_{R,\delta} \text{ and } R>R_0.
\]
Hence, $a_{R,\delta}[\varphi]\ge \Hat a_{R,\delta}[\varphi]$, where the quadratic form $\Hat a_{R,\delta}$ is defined by
\begin{equation*}
\begin{split}
&\Hat a_{R,\delta}[\varphi]=\int_{\Omega_{R,\delta}}\bigg(|\nabla \varphi|^2+ \Big( v\big(d_S(x)\big)
-\dfrac{\varepsilon}{|x|^2}\Big)|\varphi|^2\bigg)\, \ud x\,,\\
&\quad\cD(\Hat a_{R,\delta})=\{\varphi \in H^1(\Omega_{R,\delta}):\int_{\Omega_{R,\delta}} v\circ d_S\,|\varphi|^2\, \ud x < \infty\}\,,
\end{split}
\end{equation*}
one concludes that for any $R>R_0$ and $\delta\in(0,\delta_0)$ there is $C>0$ such that
	\begin{equation*}
	\mathcal{N}_{\varepsilon_0-E}(H)\leq \mathcal{N}_{\varepsilon_0-E}(\Hat A_{R,\delta})+C \text{ for any } E>0.
	\end{equation*}
By applying the same coordinate transformations as for the upper bound, i.e. first the passage to the tubular coordinates
$(r,s,t)$, and some simple lower bounds (the computations of \cite[Section~4.1]{OBP18} apply almost literally), one sees that $\mathcal{N}_{\varepsilon_0-E}(\Hat A_{R,\delta})\le \mathcal{N}_{\varepsilon_0-E}(G_{R,\delta})$
with $G_{R,\delta}$ being the operator in $L^2(\mathcal{V}_{R,\delta})$ 
\begin{equation*}
\begin{split}
&g_{R,\delta}[\varphi]:=\int_{\mathcal{V}_{R,\delta}}\bigg(|\partial_t \varphi|^2 +v(t)|\varphi|^2+|\partial_r \varphi|^2\\
&\qquad\qquad+\frac{|\partial_s \varphi|^2-\frac{\kappa^2+1+A(\delta+\varepsilon)}{4}|\varphi|^2}{r^2(1+2\delta \kappa_{\infty})^2}\bigg)\, \ud r \, \ud s \, \ud t\,,\\
&\cD(g_{R,\delta}):=\big\{\varphi \in L^2(\mathcal{V}_{R,\delta}):\partial_{r}v,\,\partial_{t}v,\,r^{-1}\partial_{s}v,\,\sqrt{|v(t)|}\varphi \in L^2(\mathcal{V}_{R,\delta}) \big\},
\end{split}
\end{equation*}
and $A>0$ is a suitable fixed constant. More precisely, there exist $\delta'>0$ and $R'>0$ such that for any $\delta \in (0,\delta')$ and $R > R'$ 
there exists $C>0$ with
\begin{equation*}
\mathcal{N}_{\varepsilon_0-E}(H)\leq \mathcal{N}_{\varepsilon_0-E}(G_{R,\delta})+C
\text{ for any } E>0.
\end{equation*}
The operator $G_{R,\delta}$ clearly commutes with $\eins\otimes\mathcal{K}_S$ where here one identifies $L^2(\mathcal{V}_{R,\delta})\simeq L^2(\mathcal{P}_{R,\delta})\otimes L^2(\mathcal{T})$, and
\begin{equation*}
G_{R,\delta}\simeq\bigoplus_{n\in\nz}G^{[n]}_{R,\delta}\,,
\end{equation*}
where $G^{[n]}_{R,\delta}$ is the operator in $L^2(\mathcal{P}_{R,\delta})$ given by its quadratic form
\begin{equation*}
\begin{split}
&g^{[n]}_{R,\delta}[\varphi]:=\int_{\mathcal{P}_{R,\delta}}\left(|\partial_t \varphi|^2+v(t)|\varphi|^2+|\partial_r \varphi|^2+\frac{\lambda_n(\mathcal{K}_S)-\frac{1+A(\delta+\varepsilon)}{4}}{r^2(1+2\delta \kappa_{\infty})^2}|\varphi|^2  \right)\, \ud r \, \ud t\,,\\
&\quad \cD(g^{[n]}_{R,\delta})=\{\varphi \in H^1(\mathcal{P}_{R,\delta}): \sqrt{|v(t)|}\varphi \in L^2(\mathcal{P}_{R,\delta}) \}\,.
\end{split}
\end{equation*}
Consequently,
\begin{equation*}
\mathcal{N}_{\varepsilon_0-E}(G_{R,\delta})=\sum_{n\in \mathbb{N}}\mathcal{N}_{\varepsilon_0-E}(G^{[n]}_{R,\delta})\ , \ E > 0\, .
\end{equation*}
By Proposition~\ref{propneu} we can assume $R$ large to have
\begin{multline*}
g^{[n]}_{R,\delta}[\varphi] \geq \int_{R}^{\infty}\int_{-\delta r}^{+ \delta r}
\bigg(
|\partial_r \varphi|^2
+
\Big(
\frac{\lambda_n(\mathcal{K}_S)-\frac{1+A(\delta+\varepsilon)}{4}}{r^2(1+2\delta \kappa_{\infty})}  + \lambda_1(H_{\delta r, N})  \Big) |\varphi|^2
\bigg)\, \ud t \, \ud r \, \\
\geq
\int_{R}^{\infty}\int_{-\delta r}^{+ \delta r}
\bigg(
|\partial_r \varphi|^2
+
\Big(
\frac{\lambda_n(\mathcal{K}_S)-\frac{1+A(\delta+\varepsilon)}{4}}{r^2(1+2\delta \kappa_{\infty})}  +
\varepsilon_0-C_0^{-1}\mathrm{e}^{-C_0\delta r}  \Big) |\varphi|^2
\bigg)\, \ud t \, \ud r.
\end{multline*}
We set 
\begin{equation*}
N:=\Big\{n \in \mathbb{N}: \ \lambda_n(\mathcal{K}_S)-\frac{1+A(\delta+\varepsilon)}{4}  \leq 0\Big\}\,,
\end{equation*}
which implies that, choosing $R$ large enough, 
\begin{equation*}
\frac{\lambda_n(\mathcal{K}_S)-\frac{1+A(\delta+\varepsilon)}{4}}{r^2(1+2\delta \kappa_{\infty})} -C_0^{-1}\mathrm{e}^{-C_0\delta r} \geq 0, \quad r > R,\quad n \geq N+1\, .
\end{equation*}
It follows that $G^{[n]}_{R,\delta} \geq \varepsilon_0$ for $n\geq N+1$ and therefore
\begin{equation*}
\mathcal{N}_{\varepsilon_0-E}(g_{R,\delta})=\sum_{n=1}^{N}\mathcal{N}_{\varepsilon_0-E}(G^{[n]}_{R,\delta})\ , \quad E > 0\ .
\end{equation*}
To treat the terms $\mathcal{N}_{\varepsilon_0-E}(G^{[n]}_{R,\delta})$ we introduce a parameter $L > 1$, denote by $m$ the integer part of $\sqrt{L}$ and write
\begin{equation*}\begin{split}
&r_p:=R+\frac{pL}{m}\,,\ t_p:=\delta r_p\,, \ p \in \{0,...,m\}\,,\ r_{m+1}:=+\infty\,,\\
&\Omega_p:=\{ (r,t) \in \mathbb{R}^2 : \ r\in (r_p,r_{p+1})\,,\ t\in (t_p,t_{p+1})\} \subset \mathcal{P}_{R,\delta}\,, \ p\in \{0,...,m\}\, ,\\
& \Omega_{m+1}:=\mathcal{P}_{R,\delta} \setminus \overline{\bigcup_{p=0}^{m} \Omega_p}\, .
\end{split}
\end{equation*}
We then introduce, for $p \in \{0,...,m \}$, the quadratic forms
\begin{equation*}
\begin{split}
&h^{[n]}_{p,\delta}[\varphi]:=\int_{\Omega_p} \left(|\partial_r \varphi|^2+|\partial_t \varphi|^2+v(t)|\varphi|^2+ \frac{\lambda_n(\mathcal{K}_S)-\frac{1+A(\delta+\varepsilon)}{4}}{r^2(1+2\delta \kappa_{\infty}) } |\varphi|^2\right)\, \ud r \, \ud t\,,\\
&\quad\cD(h^{[n]}_{p,\delta}):=\{\varphi \in H^1(\Omega_p): \ \sqrt{|v(t)|}\varphi \in L^2(\Omega_p)  \}
\end{split}
\end{equation*}
and set 
\begin{equation*}
\begin{split}
&h^{[n]}_{m+1,\delta}[\varphi]:=\int_{\Omega_{m+1}} \left(|\partial_r \varphi|^2+ v(t)|\varphi|^2+ \frac{\lambda_n(\mathcal{K}_S)-\frac{1+A(\delta+\varepsilon)}{4}}{r^2(1+2\delta \kappa_{\infty}) } |\varphi|^2\right)\, \ud r \, \ud t\,,\\
&\quad \cD(h^{[n]}_{m+1,\delta}):=\{\varphi \in H^1(\Omega_{m+1}): \ \sqrt{|v(t)|}\varphi \in L^2(\Omega_{m+1})  \}.
\end{split}
\end{equation*}
We have then
\begin{equation*}
g^{[n]}_{R,\delta}\geq \bigoplus_{p=0}^{m+1}h^{[n]}_{p,\delta},
\qquad
\mathcal{N}_{\varepsilon_0-E}(G^{[n]}_{R,\delta}) \leq \sum_{p=0}^{m+1}\mathcal{N}_{\varepsilon_0-E}(H^{[n]}_{p,\delta})\, .
\end{equation*}
Due to the assumption $(iii)$ on $v$ we can assume that $v(t)\ge\varepsilon_1$ for $(r,t)\in \Omega_{m+1}$, where
$\varepsilon_1>\varepsilon_0$ is fixed, which implies $H^{[n]}_{m+1,\delta}(v)\geq\varepsilon_0$
and 
\begin{equation*}
\mathcal{N}_{\varepsilon_0-E}(h^{[n]}_{m+1,\delta})=0\quad \text{for} \quad n\in \big\{1,...,N\big\} \quad \text{and} \quad E > 0\, .
\end{equation*}
We can now assume that $p \in \{0,...,m-1\}$. The important thing is that we can employ a separation of variables due to the definition of $\Omega_p$. We set
\begin{equation*}
\varepsilon_{p,\delta}:=\left|\frac{\lambda_1(\mathcal{K}_S)-\frac{1}{4}\big(1+A(\delta+\varepsilon)\big)}{r^2_p(1+\delta \kappa_{\infty})}   \right|
\end{equation*}
and conclude that 
\begin{equation*}
h^{[n]}_{p,\delta}[\varphi] \geq \int_{\Omega_p} \left(|\partial_r \varphi|^2+|\partial_t \varphi|^2+v(t)|\varphi|^2 \right)\, \ud r \, \ud t-\varepsilon_{p,\delta}\|\varphi\|^2_{L^2(\Omega_p)}\,.
\end{equation*}
The next goal is to show that only the term $\mathcal{N}_{\varepsilon_0-E}(h^{[n]}_{m,\delta})$ dictates the leading asymptotic behavior of $\mathcal{N}_{\varepsilon_0-E}(g^{[n]}_{m,\delta})$. For that, note that a separation of variables yields the decomposition 
\[
H^{[n]}_{p,\delta}=N_p\otimes\eins+\eins\otimes H_{t_p,N}\,,
\]
where $N_p$ is the Neumann Laplacian on $L^2(r_p,r_{p+1})$. One gets
\begin{equation*}
\mathcal{N}_{\varepsilon_0-E}(h^{[n]}_{p,\delta})\leq \mathcal{N}_{\varepsilon_0}(h^{[n]}_{p,\delta})\leq \# \Big\{(l,j) \in \mathbb{N}_0 \times \mathbb{N}: \frac{m^2\pi^2l^2}{L^2} \leq \varepsilon_0 +\varepsilon_{p,\delta}-\lambda_j(H_{t_p,N})   \Big\}\, .
\end{equation*}
Due to Corollary~\ref{corlow} we can choose $R$ so large that $\lambda_j(H_{t_p,N}) \geq \varepsilon_0+\widetilde{\varepsilon}$ for some $\widetilde{\varepsilon} > 0$ and $j \geq 2$. We now increase $R$ such that $\varepsilon_{p,\delta} < \widetilde{\varepsilon}$ for $p\in \{0,...,m-1  \}$. Therefore, 
\begin{equation*}\begin{split}
 &\# \Big\{(l,j) \in \mathbb{N}_0 \times \mathbb{N}: \frac{m^2\pi^2l^2}{L^2} \leq \varepsilon_0 +\varepsilon_{p,\delta}-\lambda_j(H_{t_p,N})   \Big\} \\
 & \qquad = \# \Big\{l \in \mathbb{N}_0: \frac{m^2\pi^2l^2}{L^2} \leq \varepsilon_0 +\varepsilon_{p,\delta}-\lambda_1(H_{t_p,N})  \Big \}
\end{split}
\end{equation*}
and the estimate of Proposition~\ref{propneu} for $\lambda_1(H_{t_p,N})$ implies
\begin{equation*}\begin{split}
\mathcal{N}_{\varepsilon_0}(h^{[n]}_{p,\delta})&\leq \# \Big\{l \in \mathbb{N}_0: \frac{m^2\pi^2l^2}{L^2} \leq \varepsilon_0 +\varepsilon_{p,\delta}-\lambda_1(H_{t_p,N})   \Big\} \\
& \leq \# \Big\{l \in \mathbb{N}_0: \frac{m^2\pi^2l^2}{L^2} \leq \varepsilon_{p,\delta}+C_0^{-1}\mathrm{e}^{-C_0t_p}   \Big\}  \\
&\leq 1+\frac{L}{\pi m}\sqrt{\varepsilon_{p,\delta}+C_0^{-1}\mathrm{e}^{-C_0t_p} }\\
& \leq 1+c^{\prime}_{R,\delta}\frac{\sqrt{L}}{r_p}\,,
\end{split}
\end{equation*}
where $c^{\prime}_{R,\delta} > 0$ is a constant independent of $n$ and $L$. Consequently, 
\begin{equation*}
\sum_{p=0}^{m-1}\mathcal{N}_{\varepsilon_0}(H^{[n]}_{p,\delta}) \leq c^{\prime \prime}_{R,\delta}\sqrt{L}\,,
\end{equation*}
where $c^{\prime \prime}_{R,\delta} > 0$ is independent of $n$ and $L$. Hence, 
\begin{equation*}
\mathcal{N}_{\varepsilon_0-E}(G^{[n]}_{R,\delta}) \leq \mathcal{N}_{\varepsilon_0-E}(H^{[n]}_{m,\delta}) +c^{\prime \prime}_{R,\delta}\sqrt{L}\ , \quad E > 0\, ,
\end{equation*}
and it remains to find a suitable upper bound to $\mathcal{N}_{\varepsilon_0-E}(H^{[n]}_{m,\delta})$. To do this we again employ a separation of variables to write
\begin{equation*}
H^{[n]}_{m,\delta}=W^{[n]}_{R,L,\delta}\otimes \mathrm{1} + \mathrm{1} \otimes H_{t_m,N}\, ,
\end{equation*}
where  $W^{[n]}_{R,L,\delta}$ is the operator associated with the quadratic form
\begin{equation*}
\begin{split}
&w^{[n]}_{R,L,\delta}[\varphi]:=\int_{R+L}^{\infty}\left(|\varphi^{\prime}|^2+ \frac{\lambda_n(\mathcal{K}_S)
-\frac{1+A(\delta+\varepsilon)}{4}}{r^2(1+2\delta \kappa_{\infty})^2}|\varphi|^2 \right)\, \ud r\,,\\
&\quad \cD(w^{[n]}_{R,L,\delta})=H^1(R+L,+\infty)\,.
\end{split}
\end{equation*}
This leads to 
\begin{equation*}
\mathcal{N}_{\varepsilon_0-E}(H^{[n]}_{m,\delta})=\# \big\{(l,j) \in \mathbb{N}\times \mathbb{N}:\ \lambda_{l}(H_{t_m,N})+\lambda_{j}(W^{[n]}_{R,L,\delta}) \leq \varepsilon_0-E   \big\}\, .
\end{equation*}
Due to the estimate 
\begin{equation*}
w^{[n]}_{R,L,\delta}[\varphi] \geq -\left|\frac{\lambda_1(\mathcal{K}_S)-\frac{1+A(\delta+\varepsilon)}{4}}{R^2(1+2\delta \kappa_{\infty})^2}  \right| \cdot \|\varphi\|^2_{L^2(R+L,+\infty)}
\end{equation*}
and Corollary~\ref{corlow} we can choose $R$ so large that
\begin{equation*}
\mathcal{N}_{\varepsilon_0-E}(H^{[n]}_{m,\delta})=\mathcal{N}_{\varepsilon_0-E-\lambda_1(H_{t_m,N})}(W^{[n]}_{R,L,\delta})\ , \quad E > 0\ .
\end{equation*}
Introducing the new variable $\rho:=(R+L)^{-1}r$ one sees that $w^{[n]}_{R,L,\delta}[\cdot]$ is unitarily equivalent
to the quadratic form $(R+L)^{-2}z^{[n]}_{\delta}$, defined on $L^2(1,+\infty)$,
\begin{equation*}
\begin{split}
&z^{[n]}_{\delta}[\varphi]:=\int_{1}^{+\infty}\left(|\varphi^{\prime}|^2+ \frac{\lambda_n(\mathcal{K}_S)-\frac{1+A(\delta+\varepsilon)}{4}}{\rho^2(1+2\delta \kappa_{\infty})^2}|\varphi|^2 \right)\, \ud \rho\,,\\
&\quad\cD(z^{[n]}_{\delta})=H^1(1,+\infty)\,.
\end{split}
\end{equation*}
Now, for $E> 0$ we set $L=L(E)=K|\ln E|$ with some constant $K> 0$ chosen shortly and denote the corresponding $m$ as $m=m(E)$. We set
\begin{equation*}
\mu(E):=(R+L(E))^2 (E-\varepsilon_0+\lambda_1(H_{\delta(R+L(E)),N}))\,,
\end{equation*}
and take into account that $\mathcal{N}_{\varepsilon_0-E}(H^{[n]}_{m,\delta})=\mathcal{N}_{-\mu(E)}(Z^{[n]}_{\delta})$, $E> 0$. Due to Proposition~\ref{propneu} we obtain
\begin{equation*}
\left|\lambda_1(H_{\delta(R+L(E)),N})-\varepsilon_0   \right| \leq C_0^{-1}\mathrm{e}^{-C_0 (R+L(E))}=C_0^{-1}\mathrm{e}^{-\delta C_0R}E^{\delta KC_0}\ .
\end{equation*}
Hence, for a sufficiently large value of $K > 0$ we conclude that 
\begin{equation*}
\mu(E)=K^2E|\ln E|^2+\orr(E |\ln E|^2)\ , \quad E \rightarrow 0^+\, .
\end{equation*}
We finally use this to obtain
\begin{equation*}\begin{split}
\limsup_{ E \rightarrow 0^+}\frac{\mathcal{N}_{\varepsilon_0-E}(g_{R,\delta}) }{|\ln E|}&=\sum_{n=1}^{N_{\delta}}\limsup_{ E \rightarrow 0^+}\frac{\mathcal{N}_{\varepsilon_0-E}(g^{[n]}_{R,\delta}) }{|\ln E|} \\
&\leq \sum_{n=1}^{N}\limsup_{ E \rightarrow 0^+}\frac{\mathcal{N}_{\varepsilon_0-E}(h^{[n]}_{R,\delta}) }{|\ln E|} +N c^{\prime \prime}_{R,\delta}\limsup_{ E \rightarrow 0^+}\frac{\sqrt{K|\ln E|}}{|\ln E|} \\
&=\sum_{n=1}^{N}\limsup_{ E \rightarrow 0^+}\frac{\mathcal{N}_{-\mu(E)}(z^{[n]}_{\delta})}{|\ln E|}  \\
&= \sum_{n=1}^{N}\limsup_{ E \rightarrow 0^+}\frac{\mathcal{N}_{-\mu(E)}(z^{[n]}_{\delta})}{|\ln \mu(E)|} \cdot \limsup_{ E \rightarrow 0^+}\frac{|\ln \mu(E)|}{|\ln E|} \\
&=\sum_{n=1}^{N}\limsup_{ E \rightarrow 0^+}\frac{\mathcal{N}_{-\mu(E)}(z^{[n]}_{\delta})}{|\ln \mu(E)|}.
\end{split}
\end{equation*}
Applying Proposition~\ref{KirschSimon} to each summand we obtain
\begin{multline*}
\limsup_{ E \rightarrow 0^+}\frac{\mathcal{N}_{\varepsilon_0-E}(H) }{|\ln E|}
\le
\limsup_{ E \rightarrow 0^+}\frac{C+ \mathcal{N}_{\varepsilon_0-E}(G_{R,\delta}) }{|\ln E|}
=
\limsup_{ E \rightarrow 0^+}\frac{\mathcal{N}_{\varepsilon_0-E}(g_{R,\delta}) }{|\ln E|}\\
\leq\frac{1}{2\pi(1+2\delta \kappa_{\infty})}\sum_{n=1}^{N}\sqrt{
\left(
\dfrac{A(\delta+\varepsilon)}{4}-\kappa_\infty\delta-\kappa_\infty^2\delta^2-\lambda_n(\mathcal{K}_S)
\right)_+}\, ,
\end{multline*}
Since $\delta$ and $\varepsilon$ can be chosen arbitrarily small, we have the result.

\section{On the essential spectrum}\label{SectionEssential}

In this section we finish the proof of Theorem~\ref{MainTheorem} characterizing the essential part of the spectrum of $H$. In a first step we observe that 
\begin{equation*}
\inf \sigma_{ess}(H)=\varepsilon_0
\end{equation*}
is a direct consequence of the results from the previous sections. 

In order to show that $[\varepsilon_0,\infty) \subset \sigma_{ess}(H)$ we construct a suitable Weyl sequence $(\varphi_n)_{n \in \mathbb{N}} \subset \cD(H)$ in tubular coordinates. Let $\chi:\mathbb{R} \rightarrow \mathbb{R}$ be a smooth function such that $\chi(x)=0$ for $x \in (-\infty,0)$ and $\chi(x)=1$ for $x \in (1,\infty)$. Let $\varphi_0$ be the ground state of the one-dimensional operator $Q$. Using the map $\Lambda$ from \eqref{LambdaMap} consider
the unitary transform 
\[
V:L^2(\Omega_{R,\delta})\to L^2({V}_{R,\delta}), \quad V \varphi= \sqrt{\det\Lambda'} \,\varphi\circ\Lambda.
\]
Then a direct computation (change of variables) shows that if $\varphi$ is supported in $\Omega_{R,\delta}$, then
\begin{multline*}
V H\varphi_n=\Big[-\frac{\partial^2}{\partial r^2}-\frac{\partial}{\partial s}\left(\frac{1}{(r+t\kappa)^2}\frac{\partial}{\partial s}  \right)-\frac{\partial^2}{\partial t^2}\\
+v(t)+w\circ \Lambda (r,s,t)
+\widehat{w}(r,s,t)  \Big]V \varphi_n,
\end{multline*}
with 
\begin{equation*}
\widehat{w}(r,s,t)  :=\frac{t\kappa^{\prime \prime}}{2(r+t\kappa)^3}-\frac{5}{4}\frac{(t\kappa^{\prime})^2}{(r+t\kappa)^4}-\frac{\kappa^2+1}{4(r+t\kappa)^2}\, .
\end{equation*}

For $n > R$ we then define
\begin{equation*}
(V \varphi_n)(r,s,t):=\varphi_0(t)\cos(kr)\chi(r-n)\chi(2n-r)\chi(t+\delta n)\chi(\delta n- t).
\end{equation*}
A direct calculation then shows that, for any $k \geq 0$, 
\begin{equation}\label{WeylCondition}\begin{split}
\lim_{n \rightarrow \infty}\frac{\|\left(H-(k^2+\varepsilon_0)\right)\varphi_n\|^2_{L^2(\mathbb{R}^3)}}{\|\varphi_n\|^2_{L^2(\mathbb{R}^3)}}&=\lim_{n \rightarrow \infty}\frac{\|\left(V H-(k^2+\varepsilon_0)\right)V\varphi_n\|^2_{L^2(\mathcal{V}_{R,\delta})}}{\|V\varphi_n\|^2_{L^2(\mathcal{V}_{R,\delta})}} \\
&=0\, ,
\end{split}
\end{equation}
which proves that $k^2+\varepsilon_0 \in \sigma(H)$. Hence, $[\varepsilon_0,+\infty)\subset \sigma(H)$.
As this set has no isolated points, the claim follows. 



\begin{thebibliography}{AGHH88}

\bibitem[AGHH]{AGHH}
S.~Albeverio, F.~Gesztesy, R.~Hoegh-Krohn, and H.~Holden, 
  \emph{Solvable models in quantum mechanics.} Second edition. With an appendix by P.Exner.
	AMS Chelsea Publishing, 2005.

\bibitem[BEL]{BEL}
J.~Behrndt, P.~Exner, and V.~Lotoreichik, \emph{Schr\"{o}dinger operators with
  {$\delta$}-interactions supported on conical surfaces}, J. Phys. A
  \textbf{47} (2014).
	
	
\bibitem[BPP]{BPP}
V. Bruneau, K. Pankrashkin, and N. Popoff, \emph{Eigenvalue counting function for Robin Laplacians on conical domains.}
J. Geom. Anal. {\bf 28} (2018) 123--151.
	
%
%


%
%


\bibitem[DOBL]{ieot}
M.~Dauge, T.~Ourmi\`eres-Bonafos, and Y.~Lafranche,
\emph{Dirichlet spectrum of the Fichera layer}, 
Integr. Equ. Operator Theory {\bf 90} (2018) 60.

\bibitem[DOBR]{DOBR}
M.~Dauge, T.~Ourmi\`eres-Bonafos, and N.~Raymond, \emph{Spectral asymptotics of
  the {D}irichlet {L}aplacian in a conical layer}, Comm. Pure Appl. Anal.
  \textbf{14} (2015), 1239--1258.

%
%



\bibitem[DEK]{dek} P. Duclos, P. Exner, and D. Krej\v{c}i\v{r}\'{\i}k, \emph{Bound states in curved quantum layers.}
 Commun. Math. Phys. {\bf 223} (2001) 13--28.


\bibitem[Ex]{Ex} P. Exner, \emph{Leaky quantum graphs: a review.}
In \emph{Analysis on graphs and its applications} (ed.
by P. Exner, J. P. Keating, P. Kuchment, T. Sunada, A. Teplyaev), Proc. Symp. Pure Math.
vol.~77, Amer. Math. Soc., Providence, RI 2008, 523--564.

\bibitem[EK]{ek} P. Exner and H. Kova\v{r}\'{\i}k, \emph{Quantum waveguides.}
Springer, 2015.

\bibitem[EL]{EL}
P.~Exner and V.~Lotoreichik, \emph{A spectral isoperimetric inequality for
  cones}, Lett. Math. Phys. \textbf{107} (2017), no.~4, 717--732.

\bibitem[E\v{S}]{es} P. Exner and P. \v{S}eba, \emph{Bound states in curved quantum waveguides.}
J. Math. Phys. {\bf 30} (1989) 2574--2580.

\bibitem[ET]{et}
P.~Exner and M.~Tater, \emph{Spectrum of {D}irichlet {L}aplacian in a conical
  layer}, J. Phys. A \textbf{43} (2010), 474023.


\bibitem[KS]{KS88}
W.~Kirsch and B.~Simon, \emph{Corrections to the classical behavior of the
  number of bound states of {S}chr\"odinger operators}, Ann.~Phys. \textbf{183}
  (1988), 122--130.

\bibitem[LOB]{LB}
V.~Lotoreichik and T.~Ourmi\`eres-Bonafos, \emph{On the bound states of
  {S}chrödinger operators with $\delta$-interactions on conical surfaces},
  Comm.~Partial Differential Equations \textbf{41} (2016) 999--1028.

\bibitem[LR]{LR} Z. Lu, J. Rowlett, \emph{On the discrete spectrum of quantum layers.}
J. Math. Phys. {\bf 53} (2012) 073519.

\bibitem[OBP]{OBP18}
T.~Ourmi\`eres-Bonafos and K.~Pankrashkin, \emph{Discrete spectrum of
  interactions concentrated near conical surfaces}, Applicable Analysis
  \textbf{97} (2018), no.~9, 1628--1649.


\bibitem[OBPP]{OBPP}
T. Ourmi\`eres-Bonafos, K. Pankrashkin, and F. Pizzichillo, \emph{Spectral asymptotics for $\delta$-interactions on sharp cones.}
J. Math. Anal. Appl. {\bf 458} (2018) 566--589.


\end{thebibliography}
\end{document}